\documentclass[12pt,a4paper,oneside]{article}
  \usepackage{authblk} 
  \usepackage{url}
  \usepackage{amsthm}
  \usepackage{amsmath}
  \usepackage{amsfonts}
  \usepackage{amssymb}
  \usepackage{amscd}
  \usepackage{authblk}
  \usepackage{bm}
  \usepackage[mathscr]{eucal}
  \usepackage{mathtools} 
  \usepackage{enumitem}
  \usepackage{physics}
  \usepackage{color}
  \usepackage{fancyhdr}
  \usepackage{hyperref}
  \usepackage{marginnote}
  \usepackage[top=2.5cm, bottom=2.7cm,left=2.5cm, right=2.5cm, marginparwidth=1.8cm]{geometry}
  \usepackage{framed}
  \usepackage[noabbrev]{cleveref} 
  \usepackage[T1]{fontenc} 
 \usepackage[utf8]{inputenc}
 \usepackage{blkarray}

\usepackage{relsize} 
  
  \numberwithin{equation}{section}
 
 \allowdisplaybreaks[1] 
  
  \theoremstyle{definition}  
   \newtheorem{defn}{Definition}[section]
   
   \newtheorem{eg}[defn]{Example}

   \newtheorem{rmk}[defn]{Remark}

  \theoremstyle{plain}  
   \newtheorem{thm}[defn]{Theorem}
   \newtheorem{lem}[defn]{Lemma}
   \newtheorem{prop}[defn]{Proposition}
   \newtheorem{cor}[defn]{Corollary}

  \theoremstyle{remark} 

   \newcommand{\B}[1]{\mathscr{B}({#1})}
   
   \newcommand{\Bt}[1]{\mathscr{B}_2({#1})}

   \newcommand{\CK}{\mathcal{K}}

   \newcommand{\BR}{\mathbb{R}}
   
   \newcommand{\BN}{\mathbb{N}}

\usepackage{centernot}
\usepackage{todonotes}
\newcommand{\nll}{\centernot{\ll}}

 \newcommand{\numberthis}{\refstepcounter{equation}\tag{\theequation}}  
\makeatletter
\newcommand\footnoteref[1]{\protected@xdef\@thefnmark{\ref{#1}}\@footnotemark}
\makeatother

\newcommand{\emailaddress}[1]{\newline{\sf#1}}

\let\tr\relax 
\DeclareMathOperator{\tr}{Tr}

\DeclareMathOperator{\ran}{Ran}

\DeclareMathOperator{\supp}{Supp}

\DeclareMathOperator{\spn}{span}

\fontsize{12}{14}

\title{Relative Entropy via  Distribution of Observables}
\author[1]{George Androulakis}
\author[2]{Tiju Cherian John}

\affil[1]{University of South Carolina, Columbia, South Carolina, USA 
\emailaddress{giorgis@math.sc.edu}}
\affil[2]{The University of Arizona, Arizona, USA
\emailaddress{tijucherian@fulbrightmail.org}}

\begin{document}
\maketitle
\tableofcontents
\newpage
 \begin{abstract}We obtain formulas for Petz-R\'enyi and Umegaki relative entropy from the idea of distribution of a positive selfadjoint operator. Classical results on R\'enyi and Kullback-Leibler divergences are applied to obtain new results and new proofs for some known results about Petz-R\'enyi and Umegaki relative entropy.   Most important among these, is a 
necessary and sufficient condition
 for the finiteness of the Petz-R\'enyi $\alpha$-relative entropy. 
All of the results presented here are valid in both finite and infinite dimensions. In particular, these results are valid for states in Fock spaces and thus are applicable to continuous variable quantum information theory.

 \end{abstract}
 \section{Introduction}
Relative entropic quantities like Petz-Rényi relative entropy and Umegaki relative entropy are used in quantum information theory to distinguish between states. The distribution of an observable is a fundamental notion in quantum probability. In this article, we prove that Petz-Rényi and Umegaki relative entropy of quantum states can be obtained using the idea of distribution of observables.
 
In a seminal article \cite{Nussbaum-Szkola-2009}, Nussbaum and Szkoła associated two probability distributions to any pair of finite dimensional quantum states and used it to study the Petz-Rényi relative entropy. These distributions are now known as Nussbaum-Szkoła distributions. In the article 
\cite{Androulakis-John-2022a}, these distributions have been generalized to infinite dimensions and were used to prove several results about quantum $f$-divergence. In the present article, we focus on two special cases of $f$-divergence: Petz-Rényi relative entropy and Umegaki  relative entropy. In Section \ref{sec:petz-renyi-umegaki} we specialize the general formula for $f$-divergence given in \cite{Androulakis-John-2022a} to the case of Petz-Rényi relative entropy and Umegaki  relative entropy. These formulas are provided in Theorem \ref{thm:quantum-reduces-to-classical}. We use these formulas in Sections  \ref{sec:petz-renyi-distribution} and \ref{sec:umegaki} to prove that the Petz-Rényi relative entropy and Umegaki  relative entropy can be obtained from the  distribution of observables. The results of this article hold in both finite and infinite dimensions. Hence our results are applicable to continuous variable quantum information theory \cite{Braunstein-van-Loock-2005, AdRaSa14, serafini2021}. In particular, Theorem \ref{thm:quantum-reduces-to-classical} of this article is used to find the precise range of $\alpha$ where the Petz-Rényi $\alpha$-relative entropy of certain class of gaussian states is finite \cite{Androulakis-John-2023}.

Another application of our formulas for Petz-Rényi relative entropy and Umegaki  relative entropy is the fact that these quantum entropic quantities (in both finite and infinite dimensions) coincide with the classical Rényi  and Kullback-Leibler divergence, respectively. This is a generalization of Nussbaum and Szkoła's result to infinite dimensions. This provides a general framework to obtain quantum results from the existing  literature on classical divergences. We illustrate this in Section \ref{sec:properties} by proving several results about Petz-Rényi and Umegaki  relative entropy using existing results about corresponding classical divergences. 

Now we describe some preliminaries and fix some notation for the rest of the article. 
Let $\CK$ be a complex Hilbert space with  $\dim \CK = \abs{\mathcal{I}}$, where $\mathcal{I} =\{1,2,\dots,n\}$ for some $n\in \BN$, or $\mathcal{I}=\BN$. Let $\rho$ and $\sigma$ be states on $\CK$  with spectral decomposition \begin{align}\label{eq:spectral-rho-sigma}\begin{split}
    \rho &= \sum_{i\in \mathcal{I}}r_i \ketbra{u_i}, \quad r_i\geq 0,\quad \sum_{i\in \mathcal{I}} r_i = 1,\quad \{u_i\}_{i\in \mathcal{I}} \text{ is an orthonormal basis of }  \mathcal{K};\\
    \sigma &= \sum_{j\in \mathcal{I}}s_j \ketbra{v_j}, \quad s_j\geq0, \quad \sum_{j\in \mathcal{I}} s_j = 1, \quad \{v_j\}_{j\in \mathcal{I}}\text{ is an orthonormal basis of } \mathcal{K}.
    \end{split}
\end{align}
The next definition of Nussbaum-Szkoła distribution is a direct generalization of the original definition in \cite{Nussbaum-Szkola-2009}.
 \begin{defn}\label{defn:nussbaum-szkola}
  (\textit{Nussbaum-Szko{\l}a distributions.})
 Define the Nussbaum-Szko{\l}a distribution $P$ and $Q$ associated with $\rho$ and $\sigma$ on $\mathcal{I}\times \mathcal{I}$ by,
\begin{align}
    \label{eq:P-and-Q}\begin{split}
  P(i,j) &= r_i\abs{\braket{u_i}{v_j}}^2,\\
  Q(i,j)& = s_j\abs{\braket{u_i}{v_j}}^2, \quad \forall (i,j)\in \mathcal{I}\times\mathcal{I}.
     \end{split}
\end{align}
 \end{defn}

  Now we define the $f$-divergence of two states $\rho$ and $\sigma$ as in \cite[Definition 2.1]{Hiai-2018}. A motivation for this definition can be seen in \cite[Equations 3.9 and 3.12]{Hiai-Mosonyi-2017}. It may be noted from the above references that the most general definition of $f$-divergence uses  Araki's relative modular operator $\Delta_{\rho,\sigma}$ \cite{araki1976relative}, but our setting only involves density operators in $\B{\CK}$. Hence, the following explicit spectral decomposition of $\Delta_{\rho,\sigma}$ may be used, \begin{align}\label{eq:delta-spectral}
    \Delta_{\rho,\sigma} = \sum_{\{i,j\,:\, r_i\neq 0, s_j \neq 0\}} r_is_j^{-1}\ketbra{X_{ij}},
\end{align}
where \begin{align}
    \label{eq:X-ell-j-defn}
    X_{ij} = \ketbra{u_i}{v_j} \in \Bt{\CK},\quad \forall i,j\in \mathcal{I}.
\end{align} A proof of \eqref{eq:delta-spectral} can be seen in \cite[Proposition B1]{Androulakis-John-2022a}.    Before defining the $f$-divergence, we need to fix a few notations and conventions. Let $\tau$ be any state on a Hilbert space $\B{\CK}$, then  $\Pi_{\tau}$ denote the orthogonal projection onto the support of $\sigma$. For a convex (or concave) function 
 $f:(0,\infty)\to \BR$, let
 \begin{align*}
     f(0) := \lim_{t\downarrow 0}f(t),&\quad f'(\infty):= \lim_{t\rightarrow\infty}\frac{f(t)}{t}.
 \end{align*}
 \begin{defn}\label{defn:f-divergence}
 Let $\rho$ and $\sigma$ be states on a Hilbert space ${\CK}$. If $f:(0,\infty)\to\BR$ is a convex (or concave) function then the $f$-divergence $D_f(\rho||\sigma)$ of $\rho$ from $\sigma$ is defined as 
 \begin{align}\label{eq:f-divergence}
     D_f(\rho||\sigma)= \int_{0^+}^{\infty}f(\lambda)\mel{\sqrt{\sigma}}{\xi^{\Delta_{\rho,\sigma}}(\dd \lambda)}{\sqrt{\sigma}}_2+f(0)\tr \left(\sigma\Pi_\rho^\perp\right)+f'(\infty)\tr\left(\rho\Pi_\sigma^\perp\right),
 \end{align}
 where $\xi^{\Delta_{\rho,\sigma}}$ denote the spectral measure associated with Araki's relative modular operator 
 $\Delta_{\rho,\sigma}$ as discussed in \eqref{eq:delta-spectral}.
 \end{defn}
 Now we state the main result in \cite{Androulakis-John-2022a}, which is crucial for the present article. 
 \begin{thm}\label{thm:f-divergence}\cite{Androulakis-John-2022a}
Let $\rho,\sigma$ be as in (\ref{eq:spectral-rho-sigma}) and $P,Q$ denote the corresponding Nussbaum-Szko{\l}a distributions.  Let $f:(0,\infty)\to\BR$ be a convex (or concave) function  and $D_f(\rho||\sigma)$, $D_{f}(P||Q)$  respectively   denote the quantum $f$-divergence of $\rho$ from $\sigma$ and the classical $f$-divergence  of  $P$ from $Q$. Then \begin{align}
    \label{eq:f-divergence-quantum-equals-classical}
    D_f(\rho||\sigma) = D_f(P||Q).
\end{align}
\end{thm}
\begin{lem}\cite{Androulakis-John-2022a}\label{lem:f-divergence-Nussbaum-Szkola-formula} The $f$-divergence of the Nussbaum-Szko{\l}a distributions can be computed as
 \begin{equation}\label{eq:f-divergence-Nussbaum-Szkola-formula}
     D_f(P||Q) 
     = \sum\limits_{\left\{\substack{i,j\,:\\ r_i s_j \neq 0} \right\}}f\left({r_i}{s_j^{-1}}\right)s_j\abs{\braket{u_i}{v_j}}^2+f(0)Q(P=0)+f'(\infty)P(Q=0).
 \end{equation}
\end{lem} 
In several occasions below, we will use the following rearrangement trick for a sum of the form $\sum_{k}f(x_k)y_k$ with $y_k>0$ for all $k$.  Notice that if the sum of the negative terms in the series above is strictly bigger than $-\infty$ (or the sum of  positive terms in the series is strictly less than $\infty$), then any rearrangement of the series produces the same sum. In  particular, for $N \in \mathbb{N} \cup \{ \infty \}$, if the sum of the negative terms in the series $ \sum\limits_{k=1}^N f(x_k)y_k$ is strictly bigger than $-\infty$, (or the sum of  its positive terms is strictly less than $\infty$), we have
\begin{equation}\label{eq:trivial-identity}
     \sum\limits_{k=1}^N f(x_k)y_k =\sum\limits_{\lambda \in \{ x_k: k=1, \ldots , N \} }f(\lambda)
     \sum\limits_{\{\ell\,:\,x_{\ell}=\lambda\}}y_{\ell}. 
 \end{equation}
 Note that in the first sum on the right side of \eqref{eq:trivial-identity}, every element $x$ in the sequence $(x_k)_{k=1}^N$ appears exactly once even if the terms $x_k$ are not distinct.

  \section{Petz-R\'enyi and Umegaki Relative Entropy of States} \label{sec:petz-renyi-umegaki}
Having defined the $f$-divergences, the most economic way to define and study other entropic quantities is through $f$-divergences. In this section, we define the Petz-R\'enyi and Umegaki relative entropy using the $f$-divergences as in \cite{Hiai-2018}. Nevertheless, this definition coincides with other definitions seen in the literature, for example Araki in \cite{araki1977relative} and Berta, Scholz and Tomamichel in \cite{ Berta-Scholz-Tomamichel-2018}.
 \begin{defn}\label{defn:alpha-relative-entropy}
 \begin{enumerate}
     \item (\textit{Petz-R\'enyi $\alpha$-relative entropy.})  For $\alpha \in (0,1)\cup (1,\infty)$, the Petz-R\'enyi $\alpha$-relative entropy of two states $\rho$ given $\sigma$ is
     \begin{align}\label{eq:1.1}
    D_{\alpha}(\rho||\sigma) = 
    \frac{1}{\alpha-1}\log D_{f_\alpha}(\rho||\sigma),
\end{align}
where \[f_\alpha(\lambda) = \lambda^\alpha,\quad \lambda\in (0,\infty).\]

\item (\textit{Umegaki relative entropy.}) The Umegaki relative entropy of $\rho$ given $\sigma$ is defined as \begin{align}\label{eq:von-neumann}
    D(\rho||\sigma)= 
    D_f(\rho||\sigma),
\end{align}
where \[f(\lambda) = \lambda\log\lambda,\quad \lambda\in (0,\infty).\]
 \end{enumerate}
\end{defn}

 Now we state a lemma and two propositions which describe some relationships between the pairs $(\rho, \sigma)$ and $(P,Q)$  described in (\ref{eq:spectral-rho-sigma}) and (\ref{eq:P-and-Q}), respectively. Reader may refer to \cite{Androulakis-John-2022a} for proofs of these results.
 \begin{lem}\label{lem:support-condition}
Let $\rho$ and $\sigma$ be as in (\ref{eq:spectral-rho-sigma}). Then  $\supp \rho \subseteq \supp \sigma$ if and only if $s_j=0$ for some $j$ implies that for every $i$ at least one of the two quantities $\{\braket{u_i}{v_j}, r_i \}$ is equal to zero.  
\end{lem}

\begin{prop}\label{prop:rho-equals-sigma-P-equals-Q}
Let $\rho$ and $\sigma$ be as in (\ref{eq:spectral-rho-sigma})  and let $P$ and $Q$ be as in (\ref{eq:P-and-Q}) then 
\[P=Q \Leftrightarrow \rho=\sigma.\]
\end{prop}

 \begin{prop}\label{prop:support-condition-iff-absolute-continuity}
Let $\rho$ and $\sigma$ be as in (\ref{eq:spectral-rho-sigma})  and let $P$ and $Q$ be as in (\ref{eq:P-and-Q}), then \[\supp \rho \subseteq \supp \sigma \Leftrightarrow P\ll Q.\]
\end{prop}
Now we proceed to prove an important result in this article. It states that the Petz-R\'enyi $\alpha$-relative entropy and the Umegaki relative entropy of two states $\rho$ and $\sigma$ are the same as the corresponding classical divergences of the Nussbaum-Szko{\l}a distributions. Also it provides a formula to compute these quantities.  Previously, the result was known for finite dimensions by the work of Nussbaum and Szko{\l}a in \cite{Nussbaum-Szkola-2009}.
    In the infinite dimensional setting, the result about Petz-R\'enyi relative entropy was only known for the special case of gauge invariant and translation invariant gaussian states and of orders of $\alpha$ in $(0,1)$ by the work of Mosonyi in
    \cite{Mosonyi-2009}. We prove it in general for all possible orders of $\alpha$ and all states (not only gaussian states) which are given via trace duality using a density operator. 
  
  Towards our main goal, we will compute the the R\'enyi divergence and the Kullback-Leibler divergence of $P$ from $Q$ 
  (Definition \ref{defn:renyi-divergence-Kullback})  in the next lemma. 
\begin{lem}\label{lem:formula-renyi-alpha-K-L} Let $\rho$ and $\sigma$ be as in \eqref{eq:spectral-rho-sigma}. Let $P$ and $Q$ denote the corresponding Nussbaum-Szko{\l}a distributions as in \eqref{eq:P-and-Q}.
\begin{enumerate}
    \item\label{item:renyi} For $\alpha\in (0,1)\cup (1,\infty)$
 the R\'enyi divergence $D_{\alpha}(P||Q)$  is given by \begin{equation}\label{eq:formula-renyi-alpha-divergence}
     D_{\alpha}(P||Q)=\frac{1}{\alpha-1} \log \sum\limits_{i,j} r_i^\alpha s_j^{1-\alpha}\abs{\braket{u_i}{v_j}}^2, 
 \end{equation}
 where for $\alpha>1$, we adopt the conventions $0^{1-\alpha} = \infty$ and $0\cdot \infty = 0$. In particular, \begin{align}\label{eq:renyi-divergence-equals-infty-1}
     D_{\alpha}(P||Q) = \infty, \textnormal{ if } P\nll Q.
 \end{align}
 \item \label{item:K-L} The Kullback-Leibler divergence $D(P||Q)$ is given by 
 \begin{equation}\label{eq:formula-kullback-leibler}
     D(P||Q)=\sum\limits_{i,j} r_i\abs{\braket{u_i}{v_j}}^2\log\left( \frac{r_i}{s_j}\right),
 \end{equation}
 where we use the conventions that $0\log (0/q) = 0$, for $q\geq0$ and $p\log(p/0)=\infty$ if $p>0$.
 In particular, \begin{align}\label{eq:kullback-leibler-equals-infty-1}
     D(P||Q) = \infty \textnormal{ if } P\nll Q.
 \end{align}
\end{enumerate}
\end{lem}
\begin{proof}
\ref{item:renyi}. By Definition \ref{defn:renyi-divergence-Kullback}, for $\alpha\in (0,1)\cup (1,\infty)$,    we take  $f_{\alpha}(\lambda) = \lambda^{\alpha},$ for $\lambda\in (0,\infty)$ to compute the R\'enyi divergence. It is enough to prove that the classical $f$-divergence satisfies
\begin{equation}\label{eq:classical-f-divergence-Nussbaum-Szkola}
    D_{f_{\alpha}}(P||Q) = \sum\limits_{i,j} r_i^\alpha s_j^{1-\alpha}\abs{\braket{u_i}{v_j}}^2,
\end{equation}
 where for $\alpha>1$, we adopt the conventions $0^{1-\alpha} = \infty$ and $0\cdot \infty = 0$.

We will use equation \eqref{eq:f-divergence-Nussbaum-Szkola-formula} to compute  $D_{f_{\alpha}}(P||Q)$.
At first we compute the second and the third terms in \eqref{eq:f-divergence-Nussbaum-Szkola-formula}. We have $f_\alpha(0)=0$, therefore, the second term \begin{align}\label{eq:2nd-term}
    f_\alpha(0)Q(P=0)=0
\end{align} in \eqref{eq:f-divergence-Nussbaum-Szkola-formula}. To compute the term $f_\alpha'(\infty)P(Q=0)$ in \eqref{eq:f-divergence-Nussbaum-Szkola-formula} we consider three cases. Case (i) $\alpha\in (0,1)$; Case (ii) $\alpha\in (1,\infty)$ and $P\ll Q$;  Case (iii) $\alpha\in (1,\infty)$ and $P\nll Q$. In Case (i), \[f_\alpha'(\infty):= \lim_{\lambda\rightarrow\infty}\frac{f_{\alpha}(\lambda)}{\lambda}=\lim_{\lambda\rightarrow\infty}\frac{1}{\lambda^{(1-\alpha)}} = 0.\] Hence $f_\alpha'(\infty)P(Q=0)=0$ in this case. In Cases (ii) and (iii) we have $f_\alpha'(\infty) = \infty$. In Case (ii) $P(Q=0)=0$ hence  $f_\alpha'(\infty)P(Q=0)=0$ in Case (ii) as well. In Case (iii),  there exists $(i,j)$ such that $Q(i,j)=0$ but $P(i,j)\neq 0$ therefore we have $f_\alpha'(\infty)P(Q=0)=\infty$ in this case. Finally, 
\begin{align}\label{eq:3rd-term}
    f_\alpha'(\infty)P(Q=0) =\begin{cases}
    0 &\text{if } \alpha\in (0,1)\\
    0& \alpha\in (1,\infty)\text{ and } P\ll Q\\
    \infty & \alpha\in (1,\infty)\text{ and } P\nll Q.
    \end{cases}
\end{align}
Now we compute the first term in \eqref{eq:f-divergence-Nussbaum-Szkola-formula}, \begin{align}\label{eq:1st-term}
   \sum\limits_{\left\{\substack{i,j\,:\\ r_i s_j \neq 0}\right\}}f_\alpha\left({r_i}{s_j^{-1}}\right)s_j\abs{\braket{u_i}{v_j}}^2  =\sum\limits_{\left\{\substack{i,j\,:\\ r_i s_j \neq 0}\right\}} r_i^\alpha s_j^{1-\alpha}\abs{\braket{u_i}{v_j}}^2 .
\end{align}
Note that if $\alpha\in (0,1)$,\begin{equation}\label{eq:1st-term-final}
    \sum\limits_{\left\{\substack{i,j\,:\\ r_i s_j \neq 0}\right\}} r_i^\alpha s_j^{1-\alpha}\abs{\braket{u_i}{v_j}}^2 = \sum\limits_{i,j} r_i^\alpha s_j^{1-\alpha}\abs{\braket{u_i}{v_j}}^2.
\end{equation}
 Also, by Proposition \ref{prop:support-condition-iff-absolute-continuity} and Lemma \ref{lem:support-condition}, the  equation above is satisfied whenever $P\ll Q$. Hence by \eqref{eq:2nd-term} and \eqref{eq:3rd-term},  we have proved \eqref{eq:classical-f-divergence-Nussbaum-Szkola} for the cases when either $\alpha\in (0,1)$, or $\alpha\in (1,\infty)$ and $P\ll Q$. If  $P\nll Q$ then there exists $(i,j)$ such that $s_j=0$ but $r_i\abs{\braket{u_i}{v_j}}^2\neq 0$, hence when $\alpha\in (1,\infty)$ under our conventions that  $0^{1-\alpha} = \infty$ and $0\cdot \infty = 0$, note that the sum  on the right side of \eqref{eq:1st-term-final} above is equal to $\infty$, even though \eqref{eq:1st-term-final} may not be valid in this case. 
 Thus when   $\alpha\in (1,\infty)$ and $P\nll Q$, by \eqref{eq:2nd-term} and \eqref{eq:3rd-term}, we have \begin{align*}
     D_{f_{\alpha}}(P||Q) &= \sum\limits_{\left\{\substack{i,j\,:\\ r_i s_j \neq 0}\right\}} r_i^\alpha s_j^{1-\alpha}\abs{\braket{u_i}{v_j}}^2+ f(0)Q(P=0)+f_{\alpha}'(\infty)P(Q=0)\\
     &=  \sum\limits_{\left\{\substack{i,j\,:\\ r_i s_j \neq 0}\right\}} r_i^\alpha s_j^{1-\alpha}\abs{\braket{u_i}{v_j}}^2+ 0+\infty\\
     & =\infty\\ &= \sum\limits_{i,j}r_i^\alpha s_j^{1-\alpha}\abs{\braket{u_i}{v_j}}^2.
 \end{align*}
 Thus the first part of the theorem is proved.
 
 \ref{item:K-L}. By Definition \ref{defn:renyi-divergence-Kullback},    we take  $f(\lambda) =\lambda\log\lambda$, for $\lambda\in (0,\infty)$ to compute the Kullback-Leibler divergence. Once again we will use equation \eqref{eq:f-divergence-Nussbaum-Szkola-formula} to compute  $D_{f}(P||Q)$. Since $\lim_{\lambda\rightarrow 0}f(\lambda) = 0$, we see that \begin{align}\label{eq:2nd-term-KL}
    f(0)Q(P=0)=0.
\end{align}
Furthermore, \eqref{eq:3rd-term} with $f_\alpha$ replaced by $f$ is satisfied in this case as well. To compute the first term in \eqref{eq:f-divergence-Nussbaum-Szkola-formula}, note that  \begin{align}\label{eq:1st-term-K-L}
   \sum\limits_{\left\{\substack{i,j\,:\\ r_i s_j \neq 0}\right\}}f\left({r_i}{s_j^{-1}}\right)s_j\abs{\braket{u_i}{v_j}}^2  &=\sum\limits_{\left\{\substack{i,j\,:\\ r_i s_j \neq 0}\right\}} r_i s_j^{-1}\left(\log r_i s_j^{-1}\right) s_j\abs{\braket{u_i}{v_j}}^2\nonumber\\
   &= \sum\limits_{\left\{\substack{i,j\,:\\ r_i s_j \neq 0}\right\}} r_i \left(\log r_i s_j^{-1}\right) \abs{\braket{u_i}{v_j}}^2.
\end{align}
Now a similar argument as in the case of R\'enyi divergence  completes the proof in this case as well.

\end{proof}
 Now we have the theorem that was promised before the previous lemma. 
  \begin{thm}\label{thm:quantum-reduces-to-classical}
 Let $\rho$ and $\sigma$ be as in (\ref{eq:spectral-rho-sigma}).  Let $D_{\alpha}(P||Q)$   and $D(P||Q)$ respectively denote the R\'enyi divergence of order $\alpha$ and the Kullback-Leibler divergence of the Nussbaum-Szko{\l}a distributions $P$ and $Q$ associated with $\rho$ and $\sigma$. Then, 
 \begin{enumerate}
     \item \label{item:quantum-reduces-classical-1} the Petz-R\'enyi $\alpha$-relative entropy of $\rho$ given $\sigma$ is equal to the R\'enyi divergence of order $\alpha$ of $P$ from $Q$, for every $\alpha\in (0,1)\cup(1,\infty)$, i.e.,
         \begin{equation}\label{eq:main-result}
  D_{\alpha}(\rho||\sigma) = D_{\alpha}(P||Q) ,\quad \forall \alpha\in (0,1)\cup(1,\infty);
\end{equation}
\item \label{item:quantum-reduces-classical-2} the Umegaki relative entropy of $\rho$ given $\sigma$ is equal to the Kullback-Leibler divergence of $P$ from $Q$, i.e.,
 \begin{equation}\label{eq:main-result-KL}
  D(\rho||\sigma) = D(P||Q).
\end{equation}
\end{enumerate}
Moreover, we have the formulae \begin{align}\label{eq:alpha-relative-entropy}
   D_{\alpha}(\rho||\sigma)&= \frac{1}{\alpha-1} \log \sum\limits_{i,j} r_i^\alpha s_j^{1-\alpha}\abs{\braket{u_i}{v_j}}^2, \quad \forall \alpha \in (0,1) \cup (1,\infty),
\end{align}
where for $\alpha>1$, we adopt the conventions $0^{1-\alpha} = \infty$ and $0\cdot \infty = 0$, and \begin{align}\label{eq:von-neumann-relative-entropy}
       D(\rho||\sigma)&= \sum\limits_{i,j} r_i\abs{\braket{u_i}{v_j}}^2\log\left( \frac{r_i}{s_j}\right),
\end{align}
with the conventions that $0\log (0/q) = 0$, for $q\geq0$ and $p\log(p/0)=\infty$ if $p>0$.
 \end{thm}
 \begin{proof}Let $f_{\alpha}$ and $f$ be as in the proof of Lemma \ref{lem:formula-renyi-alpha-K-L}. By Theorem \ref{thm:f-divergence} and Lemma \ref{lem:formula-renyi-alpha-K-L} we have \begin{align}\label{eq:formula-f-div-equals-classical}
     \begin{split}
         D_{f_\alpha}(\rho||\sigma)&=D_{f_{\alpha}}(P||Q) = \sum\limits_{i,j} r_i^\alpha s_j^{1-\alpha}\abs{\braket{u_i}{v_j}}^2,\\
     D_{f}(\rho||\sigma) &=D(P||Q)\phantom{...}=\sum\limits_{i,j} r_i\abs{\braket{u_i}{v_j}}^2\log\left( \frac{r_i}{s_j}\right),
     \end{split}
 \end{align} which complete the proof.
 \end{proof}
 \begin{rmk}\label{rmk:quamtum-classical-measurement}
In the setting of Theorem \ref{thm:quantum-reduces-to-classical} and its proof, for $i,j\in \mathcal{I}$ define $A_{ij} = \braket{u_i}{v_j}\ketbra{u_i}{v_j}$. Then $\sum_{i,j}A_{ij}A_{ij}^\dagger = I = \sum_{i,j}A_{ij}^\dagger A_{ij}$. Hence both $\{A_{ij}A_{ij}^\dagger\}$ and $\{A_{ij}^\dagger A_{ij}\}$ are POVM's. Furthermore, it may be noted that $P(i,j) = \tr \rho A_{ij}A_{ij}^\dagger$ and $Q(i,j) =  \tr \sigma A_{ij}^\dagger A_{ij}$. Thus the probability measures $P$ and $Q$ in the previous theorem are precisely those measures that are obtained by measuring $\rho$ and $\sigma$, respectively in $\{A_{ij}A_{ij}^\dagger\}$ and  $\{A_{ij}^\dagger A_{ij}\}$.
\end{rmk}
 \section{Petz-R\'enyi Relative Entropy using Pushforward Measure}\label{sec:petz-renyi-distribution}
 The idea of distribution of a quantum random variable (observable or self adjoint operator) with respect to a state is as old as quantum mechanics itself.
 In this section we exploit a slight modification of this idea, i.e.,  the \emph{pushforward of a positive compact operator with respect to a selfadjoint operator} to describe relative entropies. This terminology is motivated by the fact that the distribution of a classical random variable is simply the push forward of the probability measure with respect to the random variable. This approach has similarities with the approach of Haagerup in defining weights on noncommutative $L^p$-spaces \cite[Example 1.2 and Proposition 1.11]{Haagerup-1979-1}. More details on relative entropies in terms Haagerup's $L^p$-spaces is provided in Appendix A of \cite{Hiai-2018}. 
 \begin{defn}\label{sec:states-observables-1}\label{defn:moments} Let
  $X$ be a (possibly unbounded) selfadjoint operator defined on $D(X)\subseteq \CK$  with spectral measure $\xi^{X}$  and let $\tau$ be a bounded positive  operator on   $\CK$. Define a positive measure $\mu^{\tau,X}$ on the Borel $\sigma$-algebra $\mathcal{B}_{\mathbb{R}}$  of $\mathbb{R}$, by  \[\mu^{\tau,X}(E):= \tr \left\{\tau^{1/2}\xi^{X}(E)\tau^{1/2}\right\},\quad \forall E \in \mathcal{B}_{\mathbb{R}}.\]
  Then $\mu^{\tau,X}$ is called the \emph{pushforward  of  $\tau$ with respect to $X$}. If $\tau$ is a state, then 
  $\mu^{\tau,X}$ is a probability distribution and it is called \emph{distribution of  $X$ with respect to $\tau$}.
\end{defn}
\begin{rmk} \label{rmk:distribution}Let $\tau$ be a positive compact operator on $\CK$, with spectral decomposition
\begin{equation}
\label{eq:15}
\tau = \sum_{i}p_{i}\ketbra{u_i}{u_{i}},
\end{equation}
 where $p_{i}\geq 0$ and  $\{u_i\}$  an  orthonormal basis in $\CK$. In this case, \begin{equation}
\label{eq:9}
\mu^{\tau,X}(E) = \sum\limits_{i}^{}p_i \left\langle u_i|\xi^X(E)|u_i \right\rangle,
 \end{equation} 
 for any Borel set $E \subseteq \mathbb{R}$.  Thus $\mu^{\tau,X}$ is supported inside the spectrum of $X$. 
 Furthermore, if $X$ has a spectral decomposition of the form \[X=\sum_{j}x_j\ketbra{v_j},\] where $x_j\geq 0$ and $\{v_j\}$ is an orthonormal basis, then $\mu^{\tau,X}$ is supported on the eigenvalues $\{x_j\}$ and \begin{align}\label{eq:mu-tau-x}
     \mu^{\tau,X}\{x_j\} &=  \sum\limits_{i}^{}p_i \left\langle u_i\left\lvert\left(\sum\limits_{\{k\,:\,x_j=x_k\}}\ketbra{v_k}\right)\right\lvert u_i \right\rangle\nonumber\\
     &=\sum\limits_{i}p_i\sum\limits_{\{k\,:\,x_j=x_k\}}\abs{\braket{u_i}{v_k}}^2.
 \end{align}
\end{rmk}
When $\alpha>1$, we note that $\sigma^{(1-\alpha)}$ is defined as the  pseudo-inverse (also known as Moore-Penrose inverse) of $\sigma$ raised to the power $(\alpha-1)$.  If the spectral decomposition of $\sigma$ is $\sum_{j}s_j\ketbra{v_j}$ with $s_j\geq0$, then \begin{align}\label{eq:pseudo-inverse}
    \sigma^{(1-\alpha)} := \sum_{\{j\,:\, s_j\neq 0\}}s_j^{(1-\alpha)}\ketbra{s_j}.
\end{align}
It may be noted from the spectral theorem  \cite[Theorem 12.4]{Par12} that the pseudo-inverse as defined above is a selfadjoint operator (not necessarily bounded) because its spectral measure is supported on the  real line. Furthermore, by (\ref{eq:pseudo-inverse}), we also have $\sigma^{(1-\alpha)}$ is a positive operator.

 \begin{thm}\label{thm:Petz-renyi-equals-haagerup}
  The Petz-R\'enyi relative entropy satisfies \begin{equation}
       D_{\alpha}(\rho||\sigma)=\begin{cases} \frac{1}{\alpha-1}\log\left(\int\limits_{0}^\infty\lambda \mu^{\rho^\alpha, \sigma^{(1-\alpha)}}(\dd \lambda)\right), &\begin{array}{l}\alpha \in (0,1) \text{ or } \\ \alpha \in (1, \infty) \text{ and } \supp \rho \subseteq \supp \sigma;\end{array}\\ 
    & \\ 
       \infty, &\text{otherwise},
       \end{cases}
   \end{equation} 
   where $\mu^{\rho^\alpha, \sigma^{(1-\alpha)}}$ is the pushforward of $\rho^\alpha$   with respect to $\sigma^{(1-\alpha)}$  as in Definition \ref{sec:states-observables-1} and when $\alpha>1$, $\sigma^{(1-\alpha)}$ is taken as the pseudo-inverse of $\sigma^{(\alpha-1)}$ (see \eqref{eq:pseudo-inverse}). Consequently, \begin{equation}\label{eq:d-alpha-finite}
       D_{\alpha}(\rho||\sigma)<\infty \text{ if and only if }\int\limits_{0}^\infty\lambda \mu^{\rho^\alpha, \sigma^{(1-\alpha)}}(\dd \lambda)<\infty,
   \end{equation} 
   $\alpha \in (0,1)$ or  $\alpha \in (1, \infty)$  and $\supp \rho \subseteq \supp \sigma$.
 \end{thm}
 \begin{proof}
 Let $\rho$ and $\sigma$ be as in \eqref{eq:spectral-rho-sigma}. Then \begin{align}\label{eq:spectral-rho-sigma-alpha}\begin{split}
    \rho^\alpha &= \sum_{i\in \mathcal{I}}r_i^\alpha \ketbra{u_i}, \quad \quad\phantom{......} r_i\geq 0,\quad \sum_i r_i = 1,\quad \{u_i\}_{i} \text{ is an orthonormal basis};\\
    \sigma^{(1-\alpha)} &= \sum_{\{j\,:\,s_j\neq0\}}s_j^{1-\alpha} \ketbra{v_j}, \quad s_j>0, \quad \sum_j s_j = 1, \quad \{v_j\}_j\text{ is an orthonormal set},
    \end{split}
\end{align}
Now by putting $\tau = \rho^{\alpha}$, $X = \sigma^{(1-\alpha)}$, $x_j = s_j^{1-\alpha}$, $p_i=r_i^{\alpha}$ in \eqref{eq:mu-tau-x}, we get
\begin{align*}
    \mu^{\rho^\alpha, \sigma^{(1-\alpha)}}\{s_j^{1-\alpha}\}&= \sum\limits_{i}r_i^\alpha\sum\limits_{\{k\,:\,s_j=s_k\}}\abs{\braket{u_i}{v_k}}^2, \quad \forall j \text{ such that } s_j\neq0,\\
     \mu^{\rho^\alpha, \sigma^{(1-\alpha)}}\{0\}&=\sum\limits_{i}r_i^\alpha\sum\limits_{\{k\,:\,s_k=0\}}\abs{\braket{u_i}{v_k}}^2.
\end{align*}
Now  for $\alpha\in (0,1)\cup(1,\infty)$, using the convention $0\cdot\infty = 0$,  we avoid integrating at $0$ and obtain
\begin{align*}
    \int\limits_{0}^\infty\lambda \mu^{\rho^\alpha, \sigma^{(1-\alpha)}}(\dd \lambda) &= \sum\limits_{\{j\,:\,s_j\neq0\}}s_j^{1-\alpha}\sum\limits_{i}r_i^{\alpha}\sum\limits_{\{k\,:\,s_j=s_k\}}\abs{\braket{u_i}{v_k}}^2\\
    &=\sum\limits_{i}r_i^{\alpha}\sum\limits_{\{j\,:\,s_j\neq0\}}s_j^{1-\alpha}\sum\limits_{\{k\,:\,s_j=s_k\}}\abs{\braket{u_i}{v_k}}^2\\
    &=\sum\limits_{i}r_i^{\alpha}\sum\limits_{\{j\,:\,s_j\neq0\}}s_j^{1-\alpha}\abs{\braket{u_i}{v_j}}^2 \quad\left(\text{by } \eqref{eq:trivial-identity}  \right)\\
    &=\sum_{\{i,j\,:\,s_j\neq0\}}r_i^{\alpha}s_j^{1-\alpha}\abs{\braket{u_i}{v_j}}^2.
\end{align*}
Now, if we assume $\supp \rho\subseteq \supp \sigma$ in the case $\alpha>1$, we get \[ \int\limits_{0}^\infty\lambda \mu^{\rho^\alpha, \sigma^{(1-\alpha)}}(\dd \lambda) =\sum_{i,j}r_i^{\alpha}s_j^{1-\alpha}\abs{\braket{u_i}{v_j}}^2, \]
where for $\alpha>1$, we use Lemma \ref{lem:support-condition},  and   adopt the conventions $0^{1-\alpha} = \infty$ and $0\cdot \infty = 0$. The last sum above is same as $D_{f_\alpha}(\rho||\sigma)$ by \eqref{eq:formula-f-div-equals-classical} in the case $\alpha \in (0,1$)  or  $\alpha\in (1,\infty)$ with $\supp \rho \subseteq \supp \sigma$, where $f_\alpha (\lambda) = \lambda^\alpha$.
 \end{proof}
\begin{rmk}
In the previous theorem, we integrated the function $\lambda$ on $[0,\infty)$ using the measure $\mu^{\rho^{\alpha},\sigma^{(1-\alpha)}}$ which is supported on the the set $\{s_j^{1-\alpha}\,:\, s_j\neq 0\}\cup \{0\}$. Nevertheless the result can be obtained by integrating the function $\lambda^{1-\alpha}$ on the same interval  using the measure $\mu^{\rho^{\alpha},\sigma}$.
\end{rmk}
\section{Umegaki Relative Entropy from Distribution of Observables}\label{sec:umegaki}
A theorem similar to the following one was obtained in \cite[Theorem 20]{luczak2019relative} in a  more general setting using a different proof. In the setting of $\B{\CK}$, as we shall see now, it follows easily from the idea of distribution.   Recall that  if the spectral decomposition of an operator $X$ is $\sum_{j}x_j\ketbra{f_j}$ with $x_j\geq0$, then \begin{align}\label{eq:pseudo-log}
    -\log X := \sum_{\{j\,:\, x_j\neq 0\}}-\log x_j\ketbra{f_j},
\end{align}
where we follow Umegaki \cite[Pages 65-66]{umegaki-1962} for defining the functional calculus of the logarithm.
\begin{thm}\label{thm:von-neumann-new-formula}
The Umegaki entropy $D(\rho||\sigma)$ satisfies \begin{align}
    \label{eq:vn-entropy}
    D(\rho||\sigma) = \int\limits_{0}^{\infty} \lambda  \mu^{\rho, -\log\sigma}(\dd \lambda)- \int\limits_{0}^{\infty} \lambda  \mu^{\rho, -\log\rho}(\dd \lambda),
\end{align}
when $\supp \rho \subseteq \supp \sigma$ and at least one of the two quantities 

$\int\limits_{0^+}^{\infty} \lambda  \mu^{\rho, -\log\sigma}(\dd \lambda)$ and $\int_{0^+}^{\infty} \lambda  \mu^{\rho, -\log\rho}(\dd \lambda)$ is finite, with the conventions $\frac{0}{0}=0$, $\frac{x}{0}=\infty$ if $x>0$, and $0\cdot\infty = 0$.
\end{thm}
\begin{proof}
Let $\rho$ and $\sigma$ be as in \eqref{eq:spectral-rho-sigma}. We have by \eqref{eq:mu-tau-x}, \[\mu^{\rho,-\log \sigma}\{-\log s_j\} = \sum_{i}r_i\sum_{\{k\,:\,s_j=s_k\}}\abs{\braket{u_i}{v_k}}^2 \quad \forall j \text{ such that } s_j \neq 0, \] and $\mu^{\rho,-\log\sigma}$ is zero everywhere else except possibly at $\{0\}$. Note that when $\supp \rho \subseteq \supp \sigma$, by Lemma \ref{lem:support-condition}, equation \eqref{eq:trivial-identity} and the conventions we have, \begin{align*}
    \int\limits_{0}^{\infty} \lambda  \mu^{\rho, -\log\sigma}(\dd \lambda)
    &=\sum_{\{j\,:\,s_j\neq 0\}}\left(-\log s_j\right)\sum_{i}r_i\sum_{\{k\,:\,s_j=s_k\}}\abs{\braket{u_i}{v_k}}^2 \\
   &= \sum_{\{i,j\,:\,s_j\neq 0\}}r_i\left(-\log s_j\right)\abs{\braket{u_i}{v_j}}^2 \\
    &= \sum\limits_{i,j} -r_i\left(\log s_j\right)\abs{\braket{u_i}{v_j}}^2. \numberthis \label{eq:first-term-vn-entropy}
\end{align*}
Similar to the above situation, by \eqref{eq:mu-tau-x} we have \begin{align*}
    \mu^{\rho,-\log \rho}\{-\log r_i\} &= \sum_{j}r_j\sum_{\{k\,:\,r_i=r_k\}}\abs{\braket{u_j}{u_k}}^2\\ 
    &=\sum_{j}r_j\sum_{\{k\,:\,r_i=r_k\}}\delta_{jk} \quad (\text{where } \delta_{jk} \text{ denotes the Kronecker-}\delta \text{ function} ) \\
    &=\sum_{\{k\,:\,r_i=r_k\}}r_k, \quad \forall i \text{ such that } r_i \neq 0,\\
\end{align*} and $\mu^{\rho,-\log\rho}$ is zero everywhere else except possibly at $\{0\}$. Since $\sum_j\abs{\braket{u_i}{v_j}}^2= \norm{u_i}=1$ for all $i$, we have, \begin{align*}
    \int\limits_{0}^{\infty} \lambda  \mu^{\rho, -\log\rho}(\dd \lambda) &= \sum_{\{i\,:\,r_i\neq 0\}}\left(-\log r_i\right)\sum_{\{k\,:\,r_i=r_k\}}r_k\\
    &= \sum\limits_{\{i\,:\,r_i\neq 0\}}-r_i\log r_i \quad (\text{by } \eqref{eq:trivial-identity})\\
     &= \sum\limits_{i}-r_i\log r_i \quad (\text{since } 0\cdot \infty = 0)\\
     &=\sum\limits_{i,j}-r_i\log r_i\abs{\braket{u_i}{v_j}}^2.\numberthis \label{eq:second-term-vn-entropy}
\end{align*}
 If at least one of  the quantities $\int\limits_{0}^{\infty} \lambda  \mu^{\rho, -\log\sigma}(\dd \lambda)$ and $ \int\limits_{0}^{\infty} \lambda  \mu^{\rho, -\log\rho}(\dd \lambda)$ is  finite, we can combine the two summations in (\ref{eq:first-term-vn-entropy}) and (\ref{eq:second-term-vn-entropy}), and write  \begin{align}\label{eq:vn-entropy-formula-before-clubbing the terms}
 & \int\limits_{0}^{\infty} \lambda  \mu^{\rho, -\log\sigma}(\dd \lambda)- \int\limits_{0}^{\infty} \lambda  \mu^{\rho, -\log\rho}(\dd \lambda)\nonumber\\
 &\phantom{...........}= \sum\limits_{i,j}-r_i(\log s_j)\abs{\braket{u_i}{v_j}}^2 -\sum\limits_{i,j}-r_i(\log r_i)\abs{\braket{u_i}{v_j}}^2\nonumber\\
    &\phantom{...........}=\sum\limits_{i,j} r_i(\log r_i)\abs{\braket{u_i}{v_j}}^2-\sum\limits_{i,j}r_i(\log s_j)\abs{\braket{u_i}{v_j}}^2 \nonumber\\
    &\phantom{...........}=\sum\limits_{i,j} \left(r_i(\log r_i)\abs{\braket{u_i}{v_j}}^2-r_i(\log s_j)\abs{\braket{u_i}{v_j}}^2\right)\nonumber.
    \end{align}
    By using Lemma \ref{lem:support-condition}, and the conventions $\frac{0}{0}=0$, $\frac{x}{0}=\infty$ if $x>0$, $0\cdot\infty = 0$,   we get  
    \begin{align*}
       \int\limits_{0}^{\infty} \lambda  \mu^{\rho, -\log\sigma}(\dd \lambda)- \int\limits_{0}^{\infty} \lambda  \mu^{\rho, -\log\rho}(\dd \lambda)
    & =\sum_{i,j}r_i\abs{\braket{u_i}{v_j}}^2\log\frac{r_i}{s_j}\\
    &= \sum_{i,j}r_i\abs{\braket{u_i}{v_j}}^2\log\frac{r_i\abs{\braket{u_i}{v_j}}^2}{s_j\abs{\braket{u_i}{v_j}}^2}\\
    &= D(P||Q)\\
    &=D(\rho||\sigma)
\end{align*}where $D(P||Q)$ is the Kullback-Leibler divergence   of the Nussbaum-Szko\l a distributions $P$ and $Q$ as in \eqref{eq:formula-kullback-leibler} and we use Theorem \ref{thm:quantum-reduces-to-classical}.
\end{proof}
\begin{rmk}
Equation \eqref{eq:vn-entropy} is same as a modified version of Umegaki's definition of the relative entropy \begin{equation}
    \label{eq:araki-eq-umegaki}D(\rho||\sigma)=\tr \rho^{1/2} (\log \rho) \rho^{1/2} - \tr \rho^{1/2} (\log \sigma)\rho^{1/2},
\end{equation} whenever $\rho^{1/2}(-\log \sigma) \rho^{1/2}$ is a densely defined operator.
\end{rmk}
\section{More Properties of Petz-Rényi Relative Entropy}\label{sec:properties}
 In this section, we use the properties of classical  divergences  to prove the properties of their quantum counterparts. Some results in this section are already known with different proofs. We provide appropriate references whenever we reprove a known result. Nevertheless, our idea is to show the usefulness of Theorem \ref{thm:quantum-reduces-to-classical} by showing that a number of results about quantum entropies follow trivially from corresponding classical results.
 \subsection{Limiting Cases}
 The definition of Petz-R\'enyi $\alpha$-relative entropy excludes the values $0$, $1$ and $\infty$ of $\alpha$. Nevertheless, we can give meaning to the  entropic quantities corresponding to these values of $\alpha$ and they are important in applications too \cite{Datta2009, muller-tomamichel-etal-2013, Datta_Leditzky_2014}. 
  First we prove that the  Petz-R\'enyi $\alpha$-relative entropy is nondecreasing in $\alpha$, which will help us to extend the definition of $D_{\alpha}$ to the values $0, 1$ and $\infty$. The next Theorem is available in \cite[part 4 of Proposition 5.3]{Hiai-2018}.
\begin{thm} \label{thm:3-veh-quantum}
For $\alpha\in (0,1)\cup(1,\infty)$ the  Petz-R\'enyi entropy, $D_{\alpha}(\rho||\sigma)$ is nondecreasing in $\alpha$.
\end{thm}
\begin{proof}
This is an easy consequence of Theorem \ref{thm:quantum-reduces-to-classical} and Theorem \ref{thm:3-veh}.
\end{proof}

 Theorem \ref{thm:3-veh-quantum} enables us to extend the definition of $D_{\alpha}(\rho||\sigma)$ to the values $\alpha = 0,1$ and $\infty$ as in the following definition.
\begin{defn}\label{defn:extended-orders-quantum}
 The Petz-R\'enyi relative entropies of orders $0, 1$ and $\infty$ are defined as \begin{align*}
     D_0(\rho||\sigma) &= \lim\limits_{\alpha \downarrow 0} D_{\alpha}(\rho||\sigma),\\
     D_1(\rho||\sigma) &= \lim\limits_{\alpha \uparrow 1} D_{\alpha}(\rho||\sigma),\\
     D_{\infty}(\rho||\sigma) &= \lim\limits_{\alpha \uparrow \infty} D_{\alpha}(\rho||\sigma).
 \end{align*}
\end{defn}
With the definition above, we have the following corollary.
\begin{cor}
For $\alpha\in [0,\infty]$, the function $\alpha\mapsto D_{\alpha}(\rho||\sigma)$ is nondecreasing and thus \[D_0(\rho||\sigma)\leq D_{\alpha}(\rho||\sigma)\leq D_{\infty}(\rho||\sigma), \quad \forall \alpha\geq0.\]
\end{cor}
\begin{rmk} Following the notations used in \cite{Datta2009},
the quantity $D_0(\rho||\sigma)$  may also be written as $D_{\operatorname{min}}(\rho||\sigma)$.
\end{rmk}
 Our next result in this section shows that the Umegaki relative entropy is the limit at $1$ of Petz-R\'enyi relative entropy. It is stated or proved  in the references \cite{Hiai-2018, Berta-Scholz-Tomamichel-2018, Jaksic-Ogata-Pillet-2012, petz-ohya-1993}. We show that this result is an easy consequence of the corresponding classical fact and our Theorem \ref{thm:quantum-reduces-to-classical}.
 \begin{thm}\label{thm:5-veh-quantum}
The Umegaki relative entropy is the limit of the Petz-R\'enyi relative entropy, i.e.,
\begin{equation}\label{eq:von-neumann-equals-left-limit-1}
    D(\rho||\sigma)=\lim\limits_{\alpha \uparrow 1} D_{\alpha}(\rho||\sigma)= D_1(\rho||\sigma).
\end{equation}
Moreover, if $D(\rho||\sigma) = \infty$ or there exists $\beta>1$ such that $D_\beta (\rho||\sigma)<\infty$, then also \begin{equation}\label{eq:von-neumann-equals-right-limit-1}
    \lim\limits_{\alpha \downarrow 1} D_{\alpha}(\rho||\sigma) = D(\rho||\sigma).
\end{equation}
\end{thm}
\begin{proof}
Recall from Theorem \ref{thm:5-veh} that the R\'enyi divergence and Kullback-Leibler divergence satisfy the \eqref{eq:von-neumann-equals-left-limit-1} and \eqref{eq:von-neumann-equals-right-limit-1} with $\rho$ and $\sigma$ replaced with $P$ and $Q$ respectively, where $P$ and $Q$ are the associated Nussbaum-Szko{\l}a distributions. Now the present theorem is an easy consequence Theorem \ref{thm:quantum-reduces-to-classical}.
\end{proof}
\begin{rmk}
It is possible that $D_{\alpha}(\rho||\sigma) = \infty$ for all $\alpha>1$, but $D(\rho||\sigma)< \infty$, and hence (\ref{eq:von-neumann-equals-right-limit-1}) does not hold (See Example \ref{eg:infty-after-one} below).
\end{rmk}
Now we discuss the limits at $0$ and $\infty$.
 If $\rho = \sum_ir_i\ketbra{u_i}{u_i}$ is a spectral decomposition of $\rho$, where $\{u_i\}$ is an orthonormal basis of $\CK$, then
\begin{equation}\label{eq:support-formula}
 \supp\rho = \overline{\spn}\{u_i|r_i\neq 0 \}.   
\end{equation} 
In \cite{Datta2009}, Datta observes in the finite dimensional setting that,  \begin{align}\label{eq:datta-leditzky}
    D_0(\rho||\sigma) = -\log \tr \Pi_{\rho}\sigma,
\end{align}
where $\Pi_{\rho}$ is the projection onto the support of $\rho$, i.e., by keeping the notations as in \eqref{eq:spectral-rho-sigma} \[\Pi_{\rho} = \textnormal{Projection onto the }\spn\{u_i\,|\,r_i\neq 0\}.\] In the finite dimensions, this result  follows from our  formula (\ref{eq:alpha-relative-entropy}) as well, because  we have \begin{align}\label{eq:datta-leditzky-computation}
    \begin{split}
       D_0(\rho||\sigma)&= \lim_{\alpha\rightarrow0}\frac{1}{\alpha-1}\log \sum\limits_{i,j} r_i^{\alpha}s_j^{1-\alpha}\abs{\braket{u_i}{v_j}}^2\\
       &=\lim_{\alpha\rightarrow0}\frac{1}{\alpha-1}\log \sum\limits_{\{i,j\,:\, r_i\neq 0\}} r_i^{\alpha}s_j^{1-\alpha}\abs{\braket{u_i}{v_j}}^2\\
       & = -\log \sum\limits_{\{i,j|r_i\neq 0\}}s_j^{}\abs{\braket{u_i}{v_j}}^2\\
       &=-\log \tr \sum_{\{i|r_i\neq 0\}}\ketbra{u_i}{u_i}\sum_{j}s_j\ketbra{v_j}{v_j}\\
       &= -\log \tr \Pi_{\rho}\sigma.
    \end{split}
\end{align}
 Now we prove (\ref{eq:datta-leditzky}) in the infinite dimensional situation. A priori, the computation in (\ref{eq:datta-leditzky-computation}) cannot go through in  infinite dimensions because we need a limit theorem to pass the limit through the infinite sum. Nevertheless, Theorem \ref{thm:4-veh} helps us to prove the desired result and the following proof works in both  finite and infinite dimensional setting.
 \begin{thm}\label{thm:4-veh-quantum}
 The Petz-R\'enyi relative entropy satisfies, \begin{align}\label{eq:datta-leditzky-infinite}
    D_0(\rho||\sigma) = -\log \tr \Pi_{\rho}\sigma,
\end{align}
where $\Pi_{\rho}$ is the projection onto  $\supp \rho$.
 \end{thm}
 \begin{proof}
 Keeping the notations as in \eqref{eq:spectral-rho-sigma}, we have by Theorem \ref{thm:quantum-reduces-to-classical} and Theorem \ref{thm:4-veh},
 \begin{align*}
      D_0(\rho||\sigma) &= -\log (Q ({\{P(i,j)>0\}}))\\
      & =-\log \sum_{\{i,j|r_i>0, \braket{u_i}{v_j}\neq0\}}s_j\abs{\braket{u_i}{v_j}}^2\\
      & = -\log \sum\limits_{\{i,j|r_i\neq 0\}}s_j^{}\abs{\braket{u_i}{v_j}}^2\\
       &=-\log \tr \sum_{\{i|r_i\neq 0\}}\ketbra{u_i}{u_i}\sum_{j}s_j\ketbra{v_j}{v_j}\\
       &= -\log \tr \Pi_{\rho}\sigma.
 \end{align*}
 \end{proof}
 \begin{rmk}
 On a related note, it may be recalled that a sandwiched R\'enyi relative entropy, $\tilde{D}_{\alpha}$ was introduced independently by M\"{u}ller-Lennert et al. in \cite{muller-tomamichel-etal-2013} and Wilde et al. in \cite{wilde-winter-yang-2014}. Furthermore, Datta and Leditzky in their  \cite[Theorem 1]{Datta_Leditzky_2014} proved that \[\lim_{\alpha\rightarrow 0 }\tilde{D}_{\alpha} = D_0(\rho||\sigma),\]
 whenever $\supp \rho =\supp \sigma$.
 \end{rmk}
 \begin{thm}\label{thm:6-veh-quantum}
Let $\rho$ and $\sigma$ be as in (\ref{eq:spectral-rho-sigma}). Then 
\begin{equation}
 \label{eq:limit-infty-quantum}
 D_{\infty}(\rho||\sigma) = \log \underset{}{\sup}\left\{\frac{r_i}{s_j}\,:\,\braket{u_i}{v_j}\neq 0\right\},
\end{equation}
with the conventions that $0/0 = 0$ and $x/0 = \infty$ if $x>0$. 
 \end{thm}
 \begin{proof} Let $P$ and $Q$ be as in (\ref{eq:P-and-Q}).
By the definition  of $D_{\infty}(\rho||\sigma)$ and Theorem \ref{thm:quantum-reduces-to-classical}, we have 
\begin{align*}
 D_{\infty}(\rho||\sigma)&= \lim_{\alpha\rightarrow \infty}D_{\alpha}(\rho||\sigma)= \lim_{\alpha\rightarrow \infty}D_{\alpha}(P||Q) = D_{\infty}(P||Q).
\end{align*}
By equation (\ref{eq:limit-infty-countable-sample-space}), \begin{align*}
  D_{\infty}(P||Q) &= \log \underset{(i,j)\in \mathcal{I}\times \mathcal{I}}{\sup}\,\frac{r_i \abs{\braket{u_i}{v_j}^2}}{s_j\abs{\braket{u_i}{v_j}^2}}\\
  &=\log \underset{}{\sup}\left\{\frac{r_i}{s_j}\,:\,\braket{u_i}{v_j}\neq 0\right\}
\end{align*}
with the conventions that $0/0 = 0$ and $x/0 = \infty$ if $x>0$.
 \end{proof}
 
\subsection{Continuity, Positivity, Symmetry and Convexity} 
We begin this section with three examples which illustrate the behaviour of $D_{\alpha}(\rho||\sigma)$ when $\alpha\geq 1.$ These examples will help us to understand continuity points of $D_{\alpha}(\rho||\alpha)$.
\begin{eg}[$D_1(\rho||\sigma)<\infty$ but $D_{\alpha}(\rho||\sigma) = \infty, \forall \alpha >1$]\label{eg:infty-after-one} Let $\{u_i\}_{i=1}^{\infty}$ be any orthonormal basis on $\CK $. Take \begin{align*}
    \rho &= \sum_{i=1}^{\infty}2^{-i}\ketbra{u_i},\\
    \sigma & =s^{-1} \sum_{j=1}^{\infty}2^{-j^2}\ketbra{u_j},
\end{align*}
 where $s=\left(\sum_j 2^{-j^2}\right)^{1/2}$. In this case, keeping the notations in (\ref{eq:spectral-rho-sigma}), and  (\ref{eq:P-and-Q}), $r_i=2^{-i}$, $s_j=s^{-1}2^{-j^2}$, $\braket{u_i}{v_j}= \delta_{i,j}$, we have
 \begin{align*}
     D_1(\rho||\sigma) = D(P||Q)&= \sum_{i}r_i\log \left( \frac{r_i}{s_i}\right) = \sum_{i}2^{-i}\log \left( \frac{s2^{-i}}{2^{-i^2}}\right)\\
     &= \sum_{i}2^{-i} \log \left(s2^{-i+i^2}\right)= \sum_{i}2^{-i} \log s+ \sum_{i}2^{-i} \log \left(2^{-i+i^2}\right) \\
     &= \sum_{i}2^{-i} \log s + \sum_{i}2^{-i}(-i+i^2) <\infty.
 \end{align*}
 On the other hand, for $\alpha>1$, \begin{align*}
     \sum_{i}r_i^{\alpha}s_i^{(1-\alpha)} &=s^{-1(1-\alpha)} \sum_{i}2^{-\alpha i}2^{-(1-\alpha)i^2} = s^{-1(1-\alpha)}\sum_{i}2^{(\alpha-1)i^2-\alpha i}=\infty
 \end{align*}
 because $(\alpha-1)>0$. Therefore,
 \begin{align*}
     D_{\alpha}(\rho||\sigma) = D_{\alpha}(P||Q)&= \frac{1}{\alpha-1} \log \sum_{i}r_i^{\alpha}s_i^{1-\alpha} = \infty. 
 \end{align*}
\end{eg}
\begin{eg}[$D_{\alpha}(\rho||\sigma)<\infty$ for $1<\alpha< 2$, but $D_{2}(\rho||\sigma) = \infty$] Let $\{u_i\}_{i=1}^{\infty}$ be any orthonormal basis on $\CK $. Take \begin{align*}
    \rho &= \sum_{i=1}^{\infty}2^{-i}\ketbra{u_i},\\
    \sigma & =s^{-1} \sum_{j=1}^{\infty}2^{-2j}\ketbra{u_j},
\end{align*}
 where $s=\left(\sum_j 2^{-2j}\right)^{1/2}$. In this case, keeping the notations in (\ref{eq:spectral-rho-sigma}), and  (\ref{eq:P-and-Q}), $r_i=2^{-i}$, $s_j=s^{-1}2^{-2j}$, $\braket{u_i}{v_j}= \delta_{i,j}$. We have
for $\alpha>1$, \begin{align*}
     \sum_{i}r_i^{\alpha}s_i^{(1-\alpha)} &=s^{-1(1-\alpha)} \sum_{i}2^{-\alpha i}2^{-(1-\alpha)2i} = s^{-1(1-\alpha)}\sum_{i}2^{(\alpha-1)2i-\alpha i}= s^{-1(1-\alpha)}\sum_{i}2^{(\alpha-2)i}.
 \end{align*}
 The above series converges for $1<\alpha<2$ and diverges for $\alpha=2$. Therefore,
 \begin{align*}
     D_{\alpha}(\rho||\sigma) = D_{\alpha}(P||Q)&= \frac{1}{\alpha-1} \log \sum_{i}r_i^{\alpha}s_i^{1-\alpha}
 \end{align*}
is finite for $1<\alpha<2$ and diverges for $\alpha=2$. 
\end{eg}
\begin{eg}[$D_{2}(\rho||\sigma)<\infty$,  but $D_{\alpha}(\rho||\sigma) = \infty$ for $\alpha>2$] Let $\{u_i\}_{i=1}^{\infty}$ be any orthonormal basis on $\CK $. Take \begin{align*}
    \rho &= \sum_{i=1}^{\infty}2^{-i}\ketbra{u_i},\\
    \sigma & =s^{-1} \sum_{j=1}^{\infty}j^{2}2^{-2j}\ketbra{u_j},
\end{align*}
 where $s=\left(\sum_j j^{2}2^{-2j}\right)^{1/2}$. In this case, keeping the notations in (\ref{eq:spectral-rho-sigma}), and  (\ref{eq:P-and-Q}), $r_i=2^{-i}$, $s_j=s^{-1}j^{2}2^{-2j}$, $\braket{u_i}{v_j}= \delta_{i,j}$. We have
for $\alpha\geq 2$, \begin{align*}
     \sum_{i}r_i^{\alpha}s_i^{(1-\alpha)} &=s^{-1(1-\alpha)} \sum_{i}2^{-\alpha i}i^{2(1-\alpha)}2^{-(1-\alpha)2i}\\ 
     &= s^{-1(1-\alpha)}\sum_{i}i^{2(1-\alpha)}2^{(\alpha-1)2i-\alpha i}\\
     &= s^{-1(1-\alpha)}\sum_{i}i^{2(1-\alpha)}2^{(\alpha-2)i}.
 \end{align*}
 The above series converges for $\alpha=2$ and diverges for $\alpha>2$. Therefore,
 \begin{align*}
     D_{\alpha}(\rho||\sigma) = D_{\alpha}(P||Q)&= \frac{1}{\alpha-1} \log \sum_{i}r_i^{\alpha}s_i^{1-\alpha}
 \end{align*}
is finite for $\alpha=2$ and diverges for $\alpha>2$. 
\end{eg}
The following careful characterization of the continuity points in $\alpha$ of $D_{\alpha}(\rho||\sigma)$ in the infinite dimensions does not seem to be available in the literature.
\begin{thm}
 \label{thm:7-veh-quantum}
The Petz-R\'enyi relative entropy $D_\alpha (\rho||\sigma)$ is continuous in $\alpha$ on the set $\mathcal{A} = [0,1]\cup \{\alpha>1|\,  D_{\alpha}(\rho||\sigma)<\infty\}$.
\end{thm}
\begin{proof} By Theorem \ref{thm:quantum-reduces-to-classical} we know $D_\alpha (\rho||\sigma) = D_\alpha (P||Q)$ where $\rho$ and $\sigma$ are as in (\ref{eq:spectral-rho-sigma})  and $P$ and $Q$ are as in (\ref{eq:P-and-Q}). The result follows because the same result is true for the classical R\'enyi relative divergence (refer Theorem \ref{thm:5-veh}).
\end{proof}
The following theorem is available in
\cite[10 of Proposition 5.3]{Hiai-2018} but our proof is different.
\begin{thm}
\label{thm:8-veh-quantum} For any order $\alpha\in [0,\infty]$, 
\[D_{\alpha}(\rho||\sigma)\geq 0.\] For $\alpha>0$, $D_{\alpha}(\rho||\sigma) = 0$ if and only if $\rho=\sigma$. For $\alpha = 0$, $D_{\alpha}(\rho||\sigma) = 0$ if and only if $\supp \sigma \subseteq \supp \rho$.
\end{thm}
\begin{proof}
Follows easily from Theorem \ref{thm:8-veh} and Proposition \ref{prop:rho-equals-sigma-P-equals-Q} because of Theorem \ref{thm:quantum-reduces-to-classical}.
\end{proof}
The following Proposition is available in \cite[5 Proposition 5.3]{Hiai-2018}.
\begin{prop}
\label{prop:8-veh-quantum}
For any $0<\alpha<1$, the Petz-R\'enyi relative entropy shows the following skew-symmetry property \[D_\alpha(\rho||\sigma) = \frac{\alpha}{1-\alpha}D_{1-\alpha}(\sigma||\rho).\]
\end{prop}
\begin{proof}
Follows from Proposition \ref{prop:8-veh}.
\end{proof}
Note that in particular, Petz-R\'enyi relative entropy is symmetric for $\alpha=1/2$, and that skew-symmetry does not hold for $\alpha = 0$ and $\alpha = 1$. 
\begin{thm}
\label{thm:16-veh-quantum}
For any $0<\alpha\leq \beta <1$,\[\frac{\alpha}{\beta}\frac{1-\beta}{1-\alpha}D_{\beta}(\rho||\sigma)\leq D_\alpha (\rho||\sigma)\leq D_{\beta}(\rho||\sigma).\]
\end{thm}
\begin{proof}
Follows from Theorems \ref{thm:16-veh} and \ref{thm:quantum-reduces-to-classical}.
\end{proof}
\begin{rmk}In the light of the Theorem~\ref{thm:16-veh-quantum} we can discuss about a topology on the set of states arising from $D_{\alpha}$.
For a fixed $\alpha \in (0,1)$, one can define $\alpha$-left open ball with center $\rho$ and radius $r>0$ to be the set
$\{ \sigma\,|\, D_\alpha(\rho||\sigma) <r\}$, and subsequently define $\alpha$-left open sets to be the union of $\alpha$-left open balls.
Notice that Theorem~\ref{thm:16-veh-quantum} yields that for $\alpha, \beta \in (0,1)$, the $\alpha$-left topology is equivalent to the 
$\beta$-left topology. Similarly, one can define $\alpha$-right open balls and $\alpha$-right topologies by reversing the order 
of $\rho$ and $\sigma$ in the definition of $\alpha$-left topology. Proposition~\ref{prop:8-veh-quantum}, combined with the fact that the 
$\alpha$-left topologies are all equivalent for $0<\alpha<1$, gives that the $\alpha$-left topologies are equivalent with the 
$\beta$-right topologies for all $\alpha, \beta \in (0,1)$.
\end{rmk}
\begin{thm}
\label{thm:23-veh-quantum}The following conditions are equivalent:
\begin{enumerate}
    \item \label{item:vEH-1-quantum}$\supp \sigma \subseteq \supp \rho,$
    \item \label{item:vEH-2-quantum} $\tr \Pi_\rho\sigma = 1$, where $\Pi_\rho$ is the orthogonal projection onto $\supp \rho$.
    \item\label{item:vEH-3-quantum} $D_0(\rho||\sigma) = 0$,
    \item\label{item:vEH-4-quantum} $\lim_{\alpha \downarrow 0} D_{\alpha}(\rho||\sigma) = 0.$
\end{enumerate}
\end{thm}
\begin{proof}
\ref{item:vEH-1-quantum} $\Leftrightarrow$ \ref{item:vEH-3-quantum} follows from Theorem \ref{thm:8-veh-quantum}. 

\ref{item:vEH-3-quantum} $\Leftrightarrow$ \ref{item:vEH-2-quantum} follows from Theorem \ref{thm:4-veh-quantum}.

\ref{item:vEH-3-quantum} $\Leftrightarrow$ \ref{item:vEH-4-quantum} follows from Theorem \ref{thm:7-veh-quantum}.
\end{proof}
The following theorem, states that the Petz-R\'enyi $\alpha$-relative entropy of two states is infinity for some $\alpha\in [0,1)$ if and only if $\supp \rho\perp \supp \sigma$. This seems to be a very interesting consequence of the classical results on relative entropic quantities.
\begin{thm}
\label{thm:24-veh-quantum}The following conditions are equivalent:
\begin{enumerate}
    \item \label{item:vEH-5-quantum}$\supp \rho\perp \supp \sigma,$
    \item \label{item:vEH-6-quantum} $\tr \Pi_\rho\sigma = 0$, where $\Pi_\rho$ is the orthogonal projection onto $\supp \rho$,
    \item\label{item:vEH-7-quantum} $D_{\alpha}(\rho||\sigma) = \infty$ for some $\alpha\in [0,1)$,
    \item\label{item:vEH-8-quantum} $D_{\alpha}(\rho||\sigma) = \infty$ for all $\alpha\in [0,\infty].$
\end{enumerate}
\end{thm}
\begin{proof}
\ref{item:vEH-5-quantum}  $\Leftrightarrow$ $\supp \sigma \subseteq \ran (I-\Pi_{\rho})$ $\Leftrightarrow$ $\tr (I-\Pi_{\rho})\sigma = \tr \sigma =1$ $\Leftrightarrow$ the statement  \ref{item:vEH-6-quantum}.

By Theorem \ref{thm:4-veh-quantum}, the statement \ref{item:vEH-6-quantum} $\Leftrightarrow$ $D_0(\rho||\sigma) = \infty$  $\Leftrightarrow$ statement \ref{item:vEH-8-quantum} by Theorem \ref{thm:3-veh-quantum}.

Finally, \ref{item:vEH-7-quantum} $\Leftrightarrow$ \ref{item:vEH-8-quantum} because of Theorem \ref{thm:16-veh-quantum}.
\end{proof}
We have the following corollary by combining Theorem \ref{thm:24-veh-quantum} and  Theorem  \ref{thm:Petz-renyi-equals-haagerup}. 

\begin{cor}\begin{enumerate}
    \item $\supp \rho \not\perp \supp \sigma$ if and only if $D_{\alpha}(\rho\|\sigma)<\infty$  for all $\alpha \in [0,1)$.
    \item If $\supp \rho \not\subseteq \supp \sigma$ then $D_\alpha(\rho\|\sigma) = \infty$  for all $\alpha>1$.
    \item If $\supp \rho \subseteq \supp \sigma$ and $\alpha >1$, then  $D_\alpha(\rho\|\sigma)< \infty$ if and only if  $\int\limits_{0}^\infty\lambda \mu^{\rho^\alpha, \sigma^{(1-\alpha)}}(\dd \lambda)<\infty$.
\end{enumerate}
\end{cor}
\begin{prop}\label{cor:2-veh-quantum}
The function $[0, \infty] \ni \alpha \mapsto (\alpha-1)D_{\alpha}(\rho||\sigma)$ is convex, with the conventions that it is $0$ at $\alpha = 1$ even if $D(\rho||\sigma) = \infty$ and that it is $0$ at $\alpha = \infty$ if $\rho=\sigma$.

\end{prop}
\begin{proof}
Follows from Corollary \ref{cor:2-veh} and Theorem \ref{thm:quantum-reduces-to-classical}.
\end{proof}

\appendix
\section{ Classical Divergences}\label{appendix:classical-divergences}
In this section, we recall a few facts about the $f$-divergences and the R\'enyi divergence in the setting of classical probability. We refer to \cite{Liese-Vajda-2006} and the survey article \cite{Erven-Harremos-2014} for the following definitions and results which we state in this section. The results  from \cite{Erven-Harremos-2014} which we use in this article  are repeated here for the ease of the reader. If $\mu$ and $\nu$ are two positive measures on a measure space $(X,\mathcal{F})$, then $\nu$ is said to be absolutely continuous with respect to $
\mu$ and we write $\nu\ll \mu$, if for every $E\in \mathcal{F}$ such that $\mu (E) = 0$, then $\nu(E) = 0$.
\begin{defn}\label{defn:renyi-divergence-Kullback}
\begin{enumerate}
     \item The R\'enyi divergence of order $\alpha\in (0,1)\cup (1,\infty)$ is defined as \begin{equation}
     \label{eq:classical-renyi-alpha}
     D_{\alpha}(P||Q) = \frac{1}{\alpha-1} \log D_{f_\alpha}(P||Q), 
 \end{equation}
 where \[f_\alpha(\lambda) = \lambda^\alpha,\quad\lambda\in (0,\infty).\]
  It may be noted that the quantity \[D_{f_\alpha}(P||Q)= \int_X p^\alpha q^{1-\alpha} \dd \mu, \]
  where for $\alpha>1$, we adopt the conventions $0^{1-\alpha} = \infty$ and $0\cdot \infty = 0$.
 \item\label{defn:kullback-leibler} The Kullback-Leibler divergence of $P$ from $Q$ is defined as 
 \begin{equation}
  D(P||Q) =  D_f(P||Q) \int p\log \frac{p}{q} \dd \mu,
 \end{equation}
 where \[f(\lambda) = \lambda\log\lambda, \quad\lambda\in (0,\infty).\]
 It may be noted that
 \[ D_f(P||Q) =\int p\log \frac{p}{q} \dd \mu,\]
 with the conventions that $0\log (0/q) = 0$, for $q\geq0$ and $p\log(p/0)=\infty$ if $p>0$. Consequently, \begin{align}\label{eq:kullback-leibler-equals-infty}
     D(P||Q) = \infty,\quad \textnormal{ if }  P\nll Q.
 \end{align}
 \end{enumerate}
 \end{defn}
\begin{defn}\label{defn:extended-orders}
 The R\'enyi divergences of orders $0, 1$ and $\infty$ are defined as \begin{align*}
     D_0(P||Q) &= \lim\limits_{\alpha \downarrow 0} D_{\alpha}(P||Q)\\
     D_1(P||Q) &= \lim\limits_{\alpha \uparrow 1} D_{\alpha}(P||Q)\\
     D_{\infty}(P||Q) &= \lim\limits_{\alpha \uparrow \infty} D_{\alpha}(P||Q)
 \end{align*}
\end{defn}
The limits in Definition \ref{defn:extended-orders} always exist because R\'enyi divergence is nondecreasing in order. 
\begin{thm}\cite[Theorem 3]{Erven-Harremos-2014}. \label{thm:3-veh}
For $\alpha\in [0, \infty]$ the R\'enyi divergence $D_{\alpha}(P||Q)$ is nondecreasing in $\alpha$.
\end{thm}
\begin{thm}
 \cite[Theorem 4]{Erven-Harremos-2014}. \label{thm:4-veh}
\begin{equation}\label{eq:limit-at-0}D_0 (P||Q) = -\log (Q ({\{p>0\}})).
\end{equation}
\end{thm}   
\begin{thm} \cite[Theorem 5]{Erven-Harremos-2014}.\label{thm:5-veh}
The Kullback-Leibler divergence is the limit of the R\'enyi divergence, i.e.,
\begin{equation}\label{eq:kullback-equals-left-limit-1}
    D(P||Q)= D_1(P||Q).
\end{equation}
Moreover, if $D(P||Q) = \infty$ or there exists $\beta>1$ such that $D_\beta (P||Q)<\infty$, then also \begin{equation}\label{eq:kullback-equals-right-limit-1}
    \lim\limits_{\alpha \downarrow 1} D_{\alpha}(P||Q) = D(P||Q).
\end{equation}
\end{thm}
\begin{rmk}
It is possible that $D_{\alpha}(P||Q) = \infty$ for all $\alpha>1$, but $D(P||Q)< \infty$, and hence (\ref{eq:kullback-equals-right-limit-1}) does not hold \cite{Erven-Harremos-2014}.
\end{rmk}
For any random variable $Y$, the \emph{essential supremum of} $Y$ with respect to $P$ is $\underset{P}{\operatorname{ess\ sup}}  Y = \sup\{c|P(Y>c)>0\}$.
\begin{thm}
\cite[Theorem 6 ]{Erven-Harremos-2014}.\label{thm:6-veh}
\begin{equation}
 \label{eq:limit-infty}
 D_{\infty}(P||Q) = \log \underset{A\in \mathcal{F}}{\sup}\frac{P(A)}{Q(A)} = \log \left(\underset{P}{\operatorname{ess\ sup}} \frac{p}{q}\right),
\end{equation}
with the convention that $0/0 = 0$ and $x/0 = \infty$ if $x>0$.
\end{thm}
If the sample space $X$ is countable, then with the notational conventions of  Theorem~\ref{thm:6-veh} the 
$P$-essential supremum of $\frac{p}{q}$ reduces to the ordinary supremum of $\frac{p}{q}$, which in turn is equal to the 
supremum of $\frac{P}{Q}$, and we have 
\begin{equation}\label{eq:limit-infty-countable-sample-space}
    D_{\infty}(P||Q) = \log   \sup_x \frac{P(x)}{Q(x)},
\end{equation}
with the convention that $0/0 = 0$ and $x/0 = \infty$ if $x>0$.
\begin{thm}
 \cite[Theorem 7]{Erven-Harremos-2014}.\label{thm:7-veh}
The R\'enyi divergence $D_\alpha (P||Q)$ is continuous in $\alpha$ on $\mathcal{A} = \{\alpha\in [0,\infty]\,|\, 0\leq \alpha\leq 1\textnormal{ or } D_{\alpha}(P||Q)<\infty\}$.
\end{thm}
\begin{thm}\cite[Theorem 8]{Erven-Harremos-2014}.\label{thm:8-veh} For any order $\alpha\in [0,\infty]$ 
\[D_{\alpha}(P||Q)\geq 0.\] For $\alpha>0$, $D_{\alpha}(P||Q) = 0$ if and only if $P=Q$. For $\alpha = 0$, $D_{\alpha}(P||Q) = 0$ if and only if $Q\ll P$.
\end{thm}
\begin{prop}\cite[Proposition 2]{Erven-Harremos-2014}.\label{prop:8-veh}
For any $0<\alpha<1$, the R\'enyi divergence shows the following skew-symmetry property \[D_\alpha(P||Q) = \frac{\alpha}{1-\alpha}D_{1-\alpha}(Q||P).\]
\end{prop}
Note that in particular, R\'enyi divergence is symmetric for $\alpha=1/2$, and that skew-symmetry does not hold for $\alpha = 0$ and $\alpha = 1$.

\begin{thm}\cite[Theorem 16]{Erven-Harremos-2014}.\label{thm:16-veh}
For any $0<\alpha\leq \beta <1$,\[\frac{\alpha}{\beta}\frac{1-\beta}{1-\alpha}D_{\beta}(P||Q)\leq D_\alpha (P||Q)\leq D_{\beta}(P||Q).\]
\end{thm}

\begin{thm}\cite[Theorem 23]{Erven-Harremos-2014}.\label{thm:23-veh}The following conditions are equivalent:
\begin{enumerate}
    \item \label{item:vEH-1}$Q\ll P,$
    \item \label{item:vEH-2} $Q(\{p>0\}) = 1$,
    \item\label{item:vEH-3} $D_0(P||Q) = 0$,
    \item\label{item:vEH-4} $\lim_{\alpha \downarrow 0} D_{\alpha}(P||Q) = 0.$
\end{enumerate}
\end{thm}
\begin{thm}\cite[Theorem 24]{Erven-Harremos-2014}.\label{thm:24-veh}The following conditions are equivalent:
\begin{enumerate}
    \item \label{item:vEH-5}$P\perp Q,$
    \item \label{item:vEH-6} $Q(\{p>0\}) = 0$,
    \item\label{item:vEH-7} $D_{\alpha}(P||Q) = \infty$ for some $\alpha\in [0,1)$,
    \item\label{item:vEH-8} $D_{\alpha}(P||Q) = \infty$ for all $\alpha\in [0,\infty].$
\end{enumerate}
\end{thm}
\begin{cor}\cite[Corollary 2]{Erven-Harremos-2014}.\label{cor:2-veh}
The function $[0, \infty] \ni \alpha \mapsto (1-\alpha)D_{\alpha}(P||Q)$ is concave, with the conventions that it is $0$ at $\alpha = 1$ even if $D(P||Q) = \infty$ and that it is $0$ at $\alpha = \infty$ if $P=Q$.
\end{cor}
\section*{Acknowledgements}
We thank Mark Wilde and Mil{\'a}n Mosonyi for their comments on an earlier version \cite{androulakis2022-arxiv} of this article, and for the references that they brought to our attention. We also thank the anonymous referees for the constructive comments which improved the article. 

The second author thanks the Fulbright Scholar Program and United States-India Educational Foundation for providing funding and other support to conduct this research through a Fulbright-Nehru Postdoctoral Fellowship (Award No. 2594/FNPDR/2020), he also acknowledges the United States Army Research Office MURI award on
Quantum Network Science, awarded under grant number W911NF2110325 for partially funding this research.

\bibliographystyle{IEEEtran}
\bibliography{bibliography}
\end{document}